\documentclass{amsart}

\usepackage{graphicx}
\usepackage{bm}

\usepackage{amssymb,amsmath}
\usepackage{verbatim}

\usepackage[section]{algorithm}
\usepackage{algpseudocode}

\newtheorem{theorem}{Theorem}[section]
\newtheorem{proposition}[theorem]{Proposition}

\newtheorem{corollary}[theorem]{Corollary}

\newtheorem{definition}[theorem]{Definition}

\newtheorem{remark}[theorem]{Remark}

\newtheorem{example}[theorem]{Example}

\def\NN{{\mathbb N}}
\def\ZZ{{\mathbb Z}}

\def\KK{{\mathbb K}}

\def\Ker{{\mathrm{Ker\,}}}

\def\GF{{\mathrm{GF}}}

\def\Gr{Gr\"obner}

\def\E0{{\sc E0}}
\def\BuDDy{{\sc BuDDy}}
\def\Sylvan{{\sc Sylvan}}

\def\bX{{\bar{X}}}
\def\bR{{\bar{R}}}
\def\bI{{\bar{I}}}
\def\bJ{{\bar{J}}}

\def\bT{{\bar{\mathrm{T}}}}
\def\bS{{\bar{\mathrm{S}}}}

\def\T{{\mathrm{T}}}

\def\cC{{\mathcal{C}}}

\def\dplus{\mathop{\dot{+}}}

\def\O{{\mathcal{O}}}

\begin{document}

\title[An algebraic attack to the Bluetooth stream cipher \E0] 
{An algebraic attack to the Bluetooth stream cipher \E0}

\author[R. La Scala]{Roberto La Scala$^*$}

\author[S. Polese]{Sergio Polese$^{**}$}

\author[S.K. Tiwari]{Sharwan K. Tiwari$^{\dagger}$}

\author[A. Visconti]{Andrea Visconti$^{**}$}

\address{$^*$ Dipartimento di Matematica, Universit\`a degli Studi di Bari
``Aldo Moro'', Via Orabona 4, 70125 Bari, Italy}
\email{roberto.lascala@uniba.it}

\address{$^{**}$ CLUB -- Cryptography and Coding Theory Group,
Dipartimento di Informatica, Universit\`a degli Studi di Milano, 
Via Celoria 18, 20133 Milano, Italy}
\email{sergio.polese,andrea.visconti@unimi.it}

\address{$^{\dagger}$ Scientific Analysis Group, Defence Research
\& Development Organization, Metcalfe House, Delhi-110054, India}
\email{shrawant@gmail.com}

\thanks{
The first author acknowledges the support of Universit\`a degli Studi di Bari,
Horizon Europe Seeds, Grant ref. PANDORA - S22. The third author thanks
the Scientific Analysis Group, DRDO, Delhi, Grant ref. 1369, for the support.
The fourth author acknowledges the support of Universit\`a degli Studi di Milano,
Grant ref. PSR20.
}



\begin{abstract}
In this paper we study the security of the Bluetooth stream cipher \E0 from
the viewpoint it is a ``difference stream cipher'', that is, it is defined by a system
of explicit difference equations over the finite field $\GF(2)$. This approach
highlights some issues of the Bluetooth encryption such as the invertibility of its
state transition map, a special set of 14 bits of its 132-bit state which when
guessed implies linear equations among the other bits and finally a small
number of spurious keys, with 83 guessed bits, which are compatible with
a keystream of about 60 bits. Exploiting
these issues, we implement an algebraic attack using \Gr\ bases, SAT solvers and
Binary Decision Diagrams. Testing activities suggest that the version based on
\Gr\ bases is the best one and it is able to attack \E0 in about $2^{79}$ seconds
on an Intel i9 CPU. To the best of our knowledge, this work improves any previous
attack based on a short keystream, hence fitting with Bluetooth specifications.
\end{abstract}

\maketitle


\section{Introduction}

The Bluetooth protocol \cite{Bl} is one of the most important players in the
``wireless revolution'' of consumer electronics. This communication protocol
has started in 1999 in the mobile phones market and now it is present in almost
all mobiles, personal computers, wireless headset and speakers, remote controllers
and many other devices. Recently, the pandemic crisis has involved Bluetooth
as an excellent tool to trace close proximity contacts. Bluetooth is a secure
protocol which relies its privacy on the \E0 stream cipher. This cipher consists
of four independent Linear Feedback Shift Registers which are combined by means
of a non-linear Finite State Machine.

By assuming a ``known-plaintext attack'', the cryptanalysis of a stream cipher
is generally based on the knowledge of some amount of bits of its keystream
(see, for instance, \cite{Kl, MPS}). Following the introduction of the Bluetooth
protocol, a number of cryptanalytic results has been obtained that can be
essentially divided into two main classes: long or short keystream attacks.
Indeed, an attack is generally faster when providing a large number of keystream
bits (see, for instance, \cite{AK,FL,GBM}) but this is actually forbidden by
the Bluetooth design which has the payload of each frame associated to a single
key consisting of only 2745 bits. According to this, the cryptanalysis of \E0
in the present paper is based on a very short keystream containing about 60 bits.
Note that in addition to these data, we essentially assume the knowledge of 83 bits
of a 132-bit internal state of \E0 by brute force. To determine the remaining bits
is very fast by \Gr\ bases computations which lead to a total running time
for our attack of approximately $2^{79}$ seconds.

Another main distinction in the attacks to Bluetooth encryption is that some
of them are correlations attacks as in \cite{LMV} but one has also algebraic
cryptanalysis \cite{Ba}. In the class of algebraic attacks to \E0 or similar
stream ciphers, the most common approach involves Binary Decision Diagrams (briefly BDDs)
as in \cite{Kr,KS,SGPS,SW} and just few papers considered other solvers of polynomial
systems as \Gr\ bases and XL-algorithm \cite{AA,AF,CM}. The cryptanalysis we propose
is based on \Gr\ bases which show to be feasible solvers for polynomials systems
having a few number of solutions. Indeed, the stream cipher \E0 tends to have
few keys that are compatible with a small number of keystream bits and this
can be considered a possible flaw.

By means of a complete implementation of the proposed method, we compare
the performance of \Gr\ bases with SAT solvers and BDDs using a large
test set. Another practical algebraic attack to \E0 one has in the literature
is the BDD-based attack described in \cite{SGPS}. Our \Gr\ bases timings are much
better than the runtimes we obtain for SAT solvers and BDDs where the latter
ones confirm and improve the timings in \cite{SGPS}.

Note that Bluetooth protocol implements a reinitialization of the initial state
of \E0 using the last 128 keystream bits obtained by clocking the cipher 200 times.
Since our algebraic attack can be performed using less than 128 keystream bits,
a double initialization can be recovered simply by applying the attack twice.

The paper is structured as follows. In Section 2 we explain how to solve a polynomial
system with coefficients and solutions in a finite field by means of a guess-and-determine
strategy, that is, by using the exhaustive evaluation over the finite field of a subset of
variables and \Gr\ bases as solvers of the remaining variables. We explain that
this method is feasible especially when the system has a single or few solutions.
In Section 3 we review and expand the theory of difference stream ciphers that
has been recently introduced in \cite{LST} to cryptanalyze stream and block ciphers.
In particular, we show that the invertibility property of the explicit difference
system governing the evolution of the internal state of a difference stream cipher
is a possible issue for its security. Moreover, we study methods to eliminate
the variables of this explicit difference system in order to speed-up an algebraic
attack to the keys that are consistent with a given keystream.

In Section 4 we describe the Bluetooth stream cipher \E0 and in Section 5 and 6
we show that it is an invertible difference stream cipher by providing its
explicit difference equations together with the ones of its inverse system.
In Section 6 we explain how cryptanalytic methods for difference stream
ciphers are applied in our attack to the Bluetooth encryption. In Section 7
we present the choice of the 83 variables that are brute forced in the
guess-and-determine strategy and we explain how 14 of them have been single out
by the difference stream cipher structure of \E0 in order to speed-up
the \Gr\ bases computations.

In Section 8 we present a complete statistics of our practical algebraic attack
to \E0 by comparing \Gr\ bases with SAT solvers and BDDs. The test set we use
are $2^{17}$ random evaluations of the 83 variables for $2^3$ different keys.
The testing activity clearly shows that \Gr\ bases perform better than
the other solvers with a total running time of about $2^{79}$ seconds with
an Intel i9 processor. To the best of our knowledge, the complexity $2^{83}$
also improves any previous attack with short keystreams (see, for instance,
\cite{Kl,SGPS}). A full cryptanalysis of \E0, including its double initialization trick,
it is therefore available with complexity $2^{84}$. We end the paper with Section 9
where some conclusions are drawn.

\section{Guess-and-determine strategy}

In algebraic cryptanalysis, to attack a stream, block or public key cipher
essentially consists in solving a system of polynomial equations over a finite field
$\KK = \GF(q)$ which has generally a single or few $\KK$-solutions. Indeed, this assumption
is a natural one if a reasonable amount of data as plaintexts, ciphertexts, keystreams
and so on, is available for the attack and few keys are compatible with such data.
To fix notations, let $r,n > 0$ be two integers and consider the polynomial system
\begin{equation}
\label{sys}
\left\{
\begin{array}{ccc}
f_1 & = & 0 \\
\vdots & & \vdots \\
f_n & = & 0 \\
\end{array}
\right.
\end{equation}
where $f_i\in P = \KK[x_1,\ldots,x_r]$, for all $1\leq i\leq n$.
Denote by $J = \langle f_1,\ldots,f_n \rangle$ the ideal of $P$ generated
by the polynomials $f_i$ and consider
\[
L = \langle x_1^q - x_1, \ldots, x_r^q - x_r \rangle\subset P.
\]
The generators of the ideal $L$ are called ``field equations'' because of the following
well-known result (see, for instance, \cite{Gh}).

\begin{proposition}
Let $\bar{\KK}$ be the algebraic closure of $\KK$ and denote
\[
V(J) = \{(\alpha_1,\ldots,\alpha_r)\in \bar{\KK}^r\mid f_i(\alpha_1,\ldots,\alpha_r) = 0\
(1\leq i\leq n)\}.
\]
Put $V_\KK(J) = V(J)\cap\KK^r$ and call $V_\KK(J)$ the {\em set of the
$\KK$-solutions of $J$, that is, of the polynomial system (\ref{sys})}. We have that
$V(L) = \KK^r$ and $V_\KK(J) = V(J + L)$ where $J + L\subset P$ is a radical ideal.
\end{proposition}

An immediate upper bound to the complexity of computing $V_\KK(J)$ is clearly
\[
\lambda = q^r
\]
by assuming that the evaluation of all polynomials $f_i$ over a vector
$(\beta_1,\ldots,\beta_r)\in\KK^r$ is performed in unit time.
If $\KK = \GF(2)$, we have that SAT solvers may slightly improve such
complexity for special systems (see, for instance, \cite{Ba}, Paragraph 13.4.2.1).
Applications of the SAT solving to boolean polynomial systems arising in cryptography
are found, for example, in \cite{CB,EPV,MM}. Another approach to polynomial system
solving is symbolic computation, that is, to compute consequences of the equations
of the system ($\ref{sys}$) which allow to obtain easily its solutions. A suitable
method consists in computing a \Gr\ basis (see, for instance, \cite{GP,KR}) of the ideal
$J = \langle f_1,\ldots,f_n \rangle\subset P$. Indeed, by the Nullstellensatz Theorem
for finite fields (see \cite{Gh}) one obtains the following result.

\begin{proposition}
\label{uniqueGB}
Assume that the polynomial system $(\ref{sys})$ has a single or no $\KK$-solution.
Then, the (reduced) universal \Gr\ basis $G$ of the ideal $J + L$, that is, its \Gr\ basis
with respect to any monomial ordering of $P$ is
\[
G = 
\left\{
\begin{array}{cl}
\{x_1 - \alpha_1,\ldots,x_r - \alpha_r\} & \mbox{if}\ V_\KK(J) = \{(\alpha_1,\ldots,\alpha_r)\}, \\
\{1\} & \mbox{otherwise}.
\end{array}
\right.
\]
\end{proposition}

Observe that \Gr\ bases have generally bad exponential complexity with respect to
the number of variables $r$ which is even worse than brute force complexity $\lambda = q^r$
(see, for instance, \cite{Ba}, Section 12.2). Nevertheless, if the number of variables $r$
is moderate and the polynomial system $(\ref{sys})$ has few $\KK$-solutions then
\Gr\ bases (in general symbolic computation) become an effective tool for solving it.
Indeed, by Proposition \ref{uniqueGB} we can choose most efficient monomial orderings
as DegRevLex to solve the system. Moreover, when there are no solutions and the \Gr\
basis is simply $G = \{1\}$ then the Buchberger algorithm is stopped once a constant
in $\KK$ is obtained as an element of the current basis, say $H$. If there is a single
solution, another optimization consists in stopping the algorithm once each variable
$x_i$ ($1\leq i\leq r$) is obtained as the leading monomial of an element in $H$.

In the case $V_\KK(J)$ consists of few $\KK$-solutions, note that the cost of obtaining
them is again essentially that of computing a DegRevLex-\Gr\ basis of $J + L$.
In fact, for solving one needs to convert this basis into a Lex-\Gr\ basis by means
of the FGLM-algorithm \cite{FGLM} which has complexity $\O(r d^3)$ where $d =
\# V_\KK(J) = \dim_\KK P/(J + L)$. If the integer $d$ is small, such complexity
is dominated by the cost of computing the DegRevLex-\Gr\ basis.

To simplify notations and statements, from now on we assume that the polynomial system
($\ref{sys}$) has a single $\KK$-solution, that is, $V_\KK(J) = \{(\alpha_1,\ldots,\alpha_r)\}$.
A standard way to reduce the complexity of solving a polynomial system with a large
number of variables consists in solving equivalently many systems having less variables
that are obtained by evaluating some subset of variables, say $\{x_1,\ldots,x_s\}$
($0\leq s\leq r$), in all possible ways over the finite field $\KK$. In other words,
for all vectors $(\beta_1,\ldots,\beta_s)\in\KK^s$ one defines the linear ideal
\[
E_{\beta_1,\ldots,\beta_s} = \langle x_1 - \beta_1, \ldots, x_s - \beta_s\rangle
\subset P
\]
and consider the corresponding ideal
\[
J_{\beta_1,\ldots,\beta_s} = J + L + E_{\beta_1,\ldots,\beta_s}.
\]
Moreover, we denote by
\begin{equation*}
\begin{gathered}
H_{\beta_1,\ldots,\beta_s} = \{
f_1(\beta_1,\ldots,\beta_s,x_{s+1},\ldots,x_r),\ldots,
f_n(\beta_1,\ldots,\beta_s,x_{s+1},\ldots,x_r), \\
x_{s+1}^q - x_{s+1},\ldots,x_r^q - x_r, x_1 - \beta_1, \ldots, x_s - \beta_s \}
\subset P
\end{gathered}
\end{equation*}
the generating set of $J_{\beta_1,\ldots,\beta_s}$ and we assume that
$H_{\beta_1,\ldots,\beta_s}$ can be computed in a negligible time.
One has clearly that
\[
V(J_{\beta_1,\ldots,\beta_s}) =
\left\{
\begin{array}{cl}
\{(\alpha_1,\ldots,\alpha_r)\} & \mbox{if}\
(\beta_1,\ldots,\beta_s) = (\alpha_1,\ldots,\alpha_s), \\
\emptyset & \mbox{otherwise}.
\end{array}
\right.
\]
This approach is generally called a {\em guess-and-determine $($or hybrid \cite{BFP}$)$
strategy} which has sequential running time
\[
\mu_s = \sum_{(\beta_1,\ldots,\beta_s)\in\KK^s} \tau_{\beta_1,\ldots,\beta_s}
\]
where $\tau_{\beta_1,\ldots,\beta_s}$ is the time for computing a DegRevLex-\Gr\ basis
of $J_{\beta_1,\ldots,\beta_s}$ starting with the generating set
$H_{\beta_1,\ldots,\beta_s}$. Denote by $\tau_s$ the average runtime of such
a computation, that is
\[
\tau_s = (\sum \tau_{\beta_1,\ldots,\beta_s}) / q^s.
\]
For a sufficiently large number $0\leq s\leq r$ of evaluated variables, one has generally
that $\tau_s\leq q^{r-s}$, that is, the total running time $\mu_s = q^s \tau_s$ of a
guess-and-determine strategy improves the brute force complexity. This motivates the use
of \Gr\ bases. An optimization of this strategy is achieved for a choice of the subset
of variables $\{x_1,\ldots,x_s\}$ such that $\mu_s$ is minimal. Obviously, if $r$ is a
large number of variables the search for this optimal choice may be not a trivial task.

Observe that for all $q^s - 1$ vectors $(\beta_1,\ldots,\beta_s)\neq (\alpha_1,\ldots,\alpha_s)$
the computed \Gr\ basis is always $G = \{1\}$, that is, the average computing time $\tau_s$
is essentially obtained for inconsistent polynomial systems where \Gr\ bases behave
generally better than SAT solvers. This is another motivation for using \Gr\ bases
in a guess-and-determine strategy. The superiority of \Gr\ bases can be explained
by observing that the UNSAT case is obtained by a SAT solver essentially by exploring
the full space. On the other hand, a symbolic method proves the inconsistency
of a system of equations by constructing some consequence of them of type $1 = 0$.
For an experimental evidence of this, see for instance \cite{LST} and Section 8
of the present paper where we show that \Gr\ bases perform better than SAT solver
and Binary Decision Diagrams when attacking the stream cipher \E0.

Note finally that if $V(J_{\beta_1,\ldots,\beta_s})$ contains many solutions,
one can determine their number without actually solving as
\[
\# V(J_{\beta_1,\ldots,\beta_s}) = \dim_\KK P/J_{\beta_1,\ldots,\beta_s}.
\]
This $\KK$-dimension of the quotient algebra $P/J_{\beta_1,\ldots,\beta_s}$
can be easily computed by means of a DegRevLex-\Gr\ basis of the ideal
$J_{\beta_1,\ldots,\beta_s}$ (see, for instance, \cite{GP,KR}).

\section{Difference stream ciphers}

Aiming to apply a guess-and-determine strategy to the cryptanalysis of the Bluetooth
stream cipher \E0, we review briefly here the theory of difference stream ciphers
and their algebraic attacks. Indeed, we will show in Section 5 that \E0 is such a cipher.
For all details we refer to the recent paper \cite{LST}.

Let $\KK$ be any field and fix an integer $n > 0$. Consider a set of variables
$X(t) = \{x_1(t),\ldots,x_n(t)\}$, for all integers $t\geq 0$ and put
$X = \bigcup_{t\geq 0} X(t)$. For the corresponding polynomial algebra $R = \KK[X]$
we consider the algebra endomorphism $\sigma: R\to R$ such that $x_i(t)\mapsto x_i(t+1)$,
for all $1\leq i\leq n$ and $t\geq 0$. The algebra $R$ under the action of $\sigma$
is called the {\em algebra of (ordinary) difference polynomials}.
Fix now some integers $r_1,\ldots,r_n\geq 0$ and define the subset
\[
\bX = \{x_1(0),\ldots,x_1(r_1-1),\ldots,x_n(0),\ldots,x_n(r_n-1)\}\subset X.
\]
We finally denote by $\bR = \KK[\bX]\subset R$ the corresponding subalgebra.

\begin{definition}
Consider some polynomials $f_1,\ldots,f_n\in\bR$. A {\em system of (algebraic ordinary)
explicit difference equations} is by definition an infinite system of polynomial
equations of the kind
\begin{equation*}
\left\{
\begin{array}{ccc}
x_1(r_1 + t) & = & \sigma^t(f_1), \\
& \vdots \\
x_n(r_n + t) & = & \sigma^t(f_n). \\
\end{array}
\right.
\quad (t\geq 0)
\end{equation*}
over the infinite set of variables $X$. Such a system is denoted briefly as
\begin{equation}
\label{dsys}
\left\{
\begin{array}{ccc}
x_1(r_1) & = & f_1, \\
& \vdots \\
x_n(r_n) & = & f_n. \\
\end{array}
\right.
\end{equation}
An $n$-tuple of functions $(a_1,\ldots,a_n)$ where each $a_i:\NN\to\KK$ $(1\leq i\leq n)$
satisfies the above system is called a {\em $\KK$-solution of $(\ref{dsys})$}.
Put $r = r_1 + \ldots + r_n$ and for all $t\geq 0$ define the vector
\[
v(t) = (a_1(t),\ldots,a_1(t+r_1-1),\ldots,a_n(t),\ldots,a_n(t+r_n-1))\in\KK^r.
\]
We call $v(t)$ the {\em $t$-state of the $\KK$-solution $(a_1,\ldots,a_n)$}.
In particular, $v(0)$ is called its {\em initial state}.
\end{definition}

\begin{definition}
Consider an explicit difference system $(\ref{dsys})$. We define the algebra
endomorphism $\bT:\bR \to \bR$ by putting, for any $i = 1,2,\dots,n$
\[
x_i(0)\mapsto x_i(1),\ldots,x_i(r_i-2)\mapsto x_i(r_i-1),
x_i(r_i-1)\mapsto f_i.
\]
If $r = r_1 + \ldots + r_n$, we denote by $\T:\KK^r\to \KK^r$ the polynomial map
corresponding to $\T$. In other words, if $v(t)$ is the $t$-state
of a $\KK$-solution $(a_1,\ldots,a_n)$ we have that $\T(v(t)) = v(t+1)$, for all clocks $t\geq 0$.
We call $\bT$ the {\em state transition endomorphism} and $\T$ the {\em state transition map
of the explicit difference system $(\ref{dsys})$}.
\end{definition}

Since $v(t) = \T^t(v(0))$, it is clear that the $\KK$-solutions $(a_1,\ldots,a_n)$
of (\ref{dsys}) are in one-to-one correspondence with their initial states
\[
v(0) = (a_1(0),\ldots,a_1(r_1-1),\ldots,a_n(0),\ldots,a_n(r_n-1)).
\]

An important class of explicit difference systems are the ones such that
for any $t\leq t'$ we can compute a $t$-state by the knowledge of a $t'$-state.

\begin{definition}
Consider the state transition endomorphism $\bT:\bR\to\bR$ and the corresponding state
transition map $\T:\KK^r\to\KK^r$ of an explicit difference system $(\ref{dsys})$.
We call the system {\em invertible} if $\bT$ is an automorphism. In this case,
$\T$ is also a bijective map. 
\end{definition}

An invertibility criterion for endomorphisms of polynomial algebras can be obtained
in terms of symbolic computation and \Gr\ bases. For a comprehensive reference we refer
to the book \cite{VDE}.

\begin{theorem}
\label{invth}
Let $X = \{x_1,\ldots,x_r\}, X' = \{x'_1,\ldots,x'_r\}$ be two disjoint variable sets
and define the polynomial algebras $P = \KK[X], P' = \KK[X']$ and $Q = \KK[X\cup X'] =
P\otimes P'$. Consider an algebra endomorphism $\varphi:P\to P$ such that
$x_1\mapsto g_1,\ldots,x_r\mapsto g_r$ $(g_i\in P)$ and the corresponding ideal
$J\subset Q$ which is generated by the set $\{x'_1 - g_1,\ldots, x'_r - g_r\}$.
Moreover, we endow the polynomial algebra $Q$ by a product monomial ordering
such that $X\succ X'$. Then, the map $\varphi$ is an automorphism of $P$ if and only if
the reduced \Gr\ basis of $J$ is of the kind $\{x_1 - g'_1,\ldots,x_r - g'_r\}$
where $g'_i\in P'$, for all $1\leq i\leq r$. In this case, if $\varphi':P'\to P'$
is the algebra endomorphism such that $x'_1\mapsto g'_1, \ldots, x'_r\mapsto g'_r$
and $\xi: P\to P'$ is the isomorphism $x_1\mapsto x'_1, \ldots, x_r\mapsto x'_r$,
we have that $\xi\, \varphi^{-1}  = \varphi'\, \xi$.
\end{theorem}

Based on the above result, we introduce the following notion for invertible systems.

\begin{definition}
\label{invsys}
Denote $\bR' = \KK[\bX']$ where
\[
\bX' =
\{x'_1(0),\ldots,x'_1(r_1-1),\ldots,x'_n(0),\ldots,x'_n(r_n-1)\}
\]
and put $Q = \KK[\bX\cup\bX'] = \bR\otimes \bR'$. Consider an invertible system $(\ref{dsys})$
and the corresponding ideal $J\subset Q$ which is generated by the following polynomials,
for any $i = 1,2,\ldots,n$
\[
x'_i(0) - x_i(1),\ldots,x'_i(r_i-2) - x_i(r_i-1),x'_i(r_i-1) - f_i.
\]
Assume that $Q$ is endowed with a product monomial ordering such that $\bX\succ \bX'$
and let
\[
G = \bigcup_i
\{x_i(1) - x'_i(0),\ldots,x_i(r_i-1) - x'_i(r_i-2), x_i(0) - f'_i\}
\]
be the reduced \Gr\ basis of $J$. Denote by $g_i$ the image of $f'_i$ under the algebra
isomorphism $\bR'\to\bR$ such that, for any $i = 1,2,\ldots,n$
\[
x'_i(0)\mapsto x_i(r_i-1), x'_i(1)\mapsto x_i(r_i-2), \ldots
x'_i(r_i-1)\mapsto x_i(0).
\]
The {\em inverse of an invertible system $(\ref{dsys})$} is by definition
the following explicit difference system
\begin{equation}
\label{invdsys}
\left\{
\begin{array}{ccc}
x_1(r_1) & = & g_1, \\
& \vdots \\
x_n(r_n) & = & g_n. \\
\end{array}
\right.
\end{equation}
\end{definition}

The following results are proved in \cite{LST}.

\begin{proposition}
Let $\bT,\bS:\bR\to\bR$ be the state transition automorphisms of an invertible
system (\ref{dsys}) and its inverse system (\ref{invdsys}), respectively.
Denote by $\xi:\bR\to\bR$ the algebra automorphism such that
\[
x_i(0)\mapsto x_i(r_i-1), x_i(1)\mapsto x_i(r_i-2),\ldots,
x_i(r_i-1)\mapsto x_i(0).
\]
One has that $\xi \bS = \bT^{-1} \xi$.
\end{proposition}

By the above proposition we obtain immediately the following result
explaining how to practically reverse the evolution of the state
of an invertible system by using the corresponding inverse system.

\begin{proposition}
\label{invstat}
Let $(\ref{invdsys})$ be the inverse system of an invertible system $(\ref{dsys})$.
If $(a_1,\ldots,a_n)$ is a $\KK$-solution of $(\ref{dsys})$, consider
its $t$-state $(t\geq 0)$
\[
v = (a_1(t),\ldots,a_1(t+r_1-1),\ldots,a_n(t),\ldots,a_n(t+r_n-1)).
\]
Denote by $(b_1,\ldots,b_n)$ the $\KK$-solution of $(\ref{invdsys})$ whose
initial state is
\[
v' = (a_1(t+r_1-1),\ldots,a_1(t),\ldots,a_n(t+r_n-1),\ldots,a_n(t)).
\]
If the $t$-state of $(b_1,\ldots,b_n)$ is
\[
u' = (b_1(t),\ldots,b_1(t+r_1-1),\ldots,b_n(t),\ldots,b_n(t+r_n-1)),
\]
then the initial state of $(a_1,\ldots,a_n)$ is
\[
u = (b_1(t+r_1-1),\ldots,b_1(t),\ldots,b_n(t+r_n-1),\ldots,b_n(t)).
\]
\end{proposition}

From now on, we will assume that $\KK = \GF(q)$ is a finite field.

\begin{definition}
A {\em difference stream cipher} $\cC$ is by definition an explicit difference
system $(\ref{dsys})$ together with a polynomial $f\in\bR$. Let $(a_1,\ldots,a_n)$
be a $\KK$-solution of $(\ref{dsys})$ and denote as usual by $v(t)\in\KK^r$
$(r = r_1 + \ldots + r_n)$ its $t$-state. The initial state $v(0)$ is called
the {\em key of the $\KK$-solution $(a_1,\ldots,a_n)$} and the function $b:\NN\to\KK$
such that $b(t) = f(v(t))$ for all $t\geq 0$, is called the {\em keystream
of $(a_1,\ldots,a_n)$}. We call $f$ the {\em keystream polynomial of the cipher $\cC$}.
Finally, the cipher $\cC$ is said {\em invertible} if such is the system $(\ref{dsys})$.
\end{definition}

A ``known-plaintext attack'' to a stream cipher essentially implies the knowledge
of the keystream as the difference between the known ciphertext and plaintext streams.
Indeed, the keystream is usually provided after a sufficiently high number
of clocks in order to prevent cryptanalysis. This motivates the following notion.

\begin{definition}
\label{stream}
Let $\cC$ be a difference stream cipher consisting of the system $(\ref{dsys})$
and the keystream polynomial $f$. Let $b:\NN\to\KK$ be the keystream of a
$\KK$-solution of $(\ref{dsys})$ and fix a clock $T\geq 0$. Consider the ideal
\[
J = \sum_{t\geq T}\, \langle \sigma^t(f) - b(t) \rangle\subset R
\]
and denote by $V_\KK(J)$ the set of the $\KK$-solutions of $J$. An {\em algebraic attack
to $\cC$ by the keystream $b$ after $T$ clocks} consists in computing the $\KK$-solutions
$(a_1,\ldots,a_n)$ of the system $(\ref{dsys})$ such that $(a_1,\ldots,a_n)\in V_\KK(J)$.
In other words, by considering the ideal corresponding to $(\ref{dsys})$, that is
\[
I = \sum_{t\geq 0}\, \langle x_1(r_1+t) - \sigma^t(f_1),\ldots,x_n(r_n+t) - \sigma^t(f_n)
\rangle \subset R
\]
we want to compute $V_\KK(I + J) = V_\KK(I)\cap V_\KK(J)$.
\end{definition}

Since the given function $b$ is the keystream of a $\KK$-solution of (\ref{dsys}),
say $(a_1,\ldots,a_n)$, we have clearly that $(a_1,\ldots,a_n)\in V_\KK(I + J)\neq\emptyset$.
For actual ciphers, we have generally that $V_\KK(I + J) = \{(a_1,\ldots,a_n)\}$.
We will assume such a unique solution from now on.

In practice, a finite number of values of the keystream $b$ is actually provided in algebraic
attacks. In other words, for a fixed integer bound $B\geq T$, we consider the polynomial algebra
$R_B = \KK[X_B]$ over the finite variable set $X_B = \bigcup_{0\leq t\leq B} X(t)$
and the ideals $I_B,J_B\subset R_B$ whose finite generating sets are respectively
\begin{equation*}
\begin{gathered}
\{x_i(r_i+t) - \sigma^t(f_i)\in R_B\mid 1\leq i\leq n,t\geq 0\}, \\
\{\sigma^t(f) - b(t)\in R_B\mid t\geq T\}.
\end{gathered}
\end{equation*}
It is shown in \cite{LST} that for a sufficiently large bound $B$, the uniqueness
of the $\KK$-solution is preserved and we have that
\[
V_\KK(I_B + J_B) = \{(a'_1,\ldots,a'_n)\}
\]
where each function $a'_i:\{0,\ldots,B\}\to\KK$ ($1\leq i\leq n$) is such that $a'_i(t) = a_i(t)$
for all $0\leq t\leq B$. Since $\KK$-solutions are in one-to-one correspondence with their
initial states that are the keys of a difference stream cipher, this means that
we can perform an actual algebraic attack without obtaining any spurious, that is,
incorrect key once a sufficiently large number of keystream values is provided.

If the explicit difference system (\ref{dsys}) is invertible we can essentially
assume that the initial keystream clock is $T = 0$. In fact, by means of the notion
of inverse system in Definition \ref{invsys} the computation of the $T$-state
is completely equivalent to the computation of the initial state, that is, the key.
Moreover, if we consider the ideal
\[
J' = \sum_{t\geq 0}\, \langle \sigma^t(f) - b(T+t) \rangle\subset R
\]
we have that the initial states of the $\KK$-solutions in $V_\KK(I + J')$ are exactly
the $T$-states of the $\KK$-solutions in $V_\KK(I + J)$. By putting $B' = B - T$,
an algebraic attack to an invertible difference stream cipher is therefore reduced to
the computation of $V_\KK(I_{B'} + J'_{B'})$. Since $T$ is generally an high value
of clock, for invertible ciphers this is a very effective optimization of the algebraic
cryptanalysis because this trick reduces drastically the number of variables to solve.
Indeed, instead of solving equations in the polynomial algebra $R_B$ ($0\leq T\leq B$)
we can solve equivalent equations in $R_{B'}$ ($B' = B - T$).


Another possible option to reduce the number of variables may consist in eliminating
all variables except for the initial variables, that is, the variables of the set $\bX$
by means of the explicit difference equations of the system (\ref{dsys}). In this way
one obtains the equations satisfied by the keys of a difference stream cipher having
a given keystream. To explain this, we introduce the following notions (see also
\cite{GLS,LS}).

\begin{definition}
A commutative algebra $A$ together with an algebra endomorphism $\alpha:A\to A$ is called
a {\em difference algebra}. An ideal $J\subset A$ such that $\alpha(J)\subset J$ is said
a {\em difference ideal}. Let $(A,\alpha), (B,\beta)$ be two difference algebras.
An algebra homomorphism $\varphi:A\to B$ is called a {\em difference algebra homomorphism}
if $\varphi \alpha = \beta \varphi$.
\end{definition}

Note that the kernel of a difference algebra homomorphism is clearly a difference ideal.
In our context, we consider the difference algebras $(R,\sigma)$ and $(\bR,\bT)$. The ideal
$I = \sum_{t\geq 0}\, \langle x_1(r_1+t) - \sigma^t(f_1),\ldots,x_n(r_n+t) - \sigma^t(f_n) \rangle$
is a difference ideal of $R$. We denote hence
$I = \langle x_1(r_1) - f_1,\ldots,x_n(r_n) - f_n \rangle_\sigma$ to mean that the set
$\{x_1(r_1) - f_1,\ldots,x_n(r_n) - f_n\}$ generates $I$ as a difference ideal.

\begin{theorem}
Let $\varphi:R\to\bR$ be the algebra homomorphism such that its restriction to $\bR$
is the identity map and for all $t\geq 0$ one has that
\[
x_i(r_i + t)\mapsto \bT^t(f_i)\ (1\leq i\leq n).
\]
It holds that $\varphi$ is a difference algebra homomorphism and its kernel is $I = \ker\varphi$.
\end{theorem}

\begin{proof}
Since $\varphi$ is an algebra homomorphism, it is sufficient to show that the property
$\varphi \sigma = \bT \varphi$ holds over the variables. For the variables $x_i(0),\ldots,x_i(r_i-2)$
($1\leq i\leq n$) this is trivial since the actions of $\sigma,\bT$ coincide and the restriction
of $\varphi$ to $\bR$ is the identity map. Moreover, we have
\[
\varphi(\sigma(x_i(r_i-1)) = \varphi(x_i(r_i)) = f_i = \bT(x_i(r_i-1)) = \bT(\varphi(x_i(r_i-1))).
\]
In the same way, for all $t > 0$ one has
\[
\varphi(\sigma(x_i(r_i-1 + t)) = \varphi(x_i(r_i + t)) = \bT^t(f_i) = \bT\bT^{t-1}(f_i)
= \bT(\varphi(x_i(r_i-1 + t))).
\]
We show now that $\Ker\varphi = I$. Since $I =
\langle x_1(r_1) - f_1,\ldots,x_n(r_n) - f_n \rangle_\sigma$
and $\varphi(x_i(r_i) - f_i) = f_i - f_i = 0$, we have that $I\subset\Ker\varphi$.

Consider now a polynomial $f\in R$. By means of the identities
$x_i(r_i + t)\equiv \sigma^t(f_i)$ modulo $I$ ($1\leq i\leq n, t\geq 0$) we have that
$f\equiv \bar{f}$ modulo $I$, for some polynomial $\bar{f}\in \bR$. Assume now that
$f\in\Ker\varphi$. Since $I\subset\Ker\varphi$, we have that $\bar{f}\in\Ker\varphi$.
Because the restriction of $\varphi$ to $\bR$ is the identity map, we have that
$\bar{f} = 0$ and hence $f\in I$.
\end{proof}

Denote by $\sigma':R/I\to R/I$ the algebra endomorphism which is induced by $\sigma:R\to R$
since $\sigma(I)\subset I$. By the above result we obtain immediately the following one.

\begin{corollary}
The difference algebras $(R/I,\sigma')$ and $(\bR,\bT)$ are isomorphic by means
of the difference algebra isomorphism $\varphi':R/I\to\bR$ which is induced by $\varphi$.
\end{corollary}

In terms of bijective polynomial maps corresponding to algebra isomorphisms, the above result
can be restated in the following way.

\begin{proposition}
Consider $V_\KK(I)$ the set of all $\KK$-solutions of $(\ref{dsys})$.
The map $\iota:V_\KK(I)\to\KK^r$ $(r = r_1 + \ldots + r_n)$ such that
\[
(a_1,\ldots,a_n)\mapsto (a_1(0),\ldots,a_1(r_1-1),\ldots,a_n(0),\ldots,a_n(r_n-1))
\]
is bijective and $\iota,\iota^{-1}$ are both polynomial maps.
\end{proposition}

From the above result we obtain that the set of keys that are compatible with
the given keystream $b$, that is, $\iota(V_\KK(I + J))$ is indeed the set
of the $\KK$-solutions of some polynomial system in $\bR$. In other words,
we have that $\iota(V_\KK(I + J)) = V_\KK(\bJ)$ for some ideal $\bJ\subset\bR$.
We look for a generating set of $\bJ$.

\begin{proposition}
We have that
\[
\bJ = \sum_{t\geq T} \langle \bT^t(f) - b(t) \rangle.
\]
\end{proposition}

\begin{proof}
Since $f\in\bR$, it is sufficient to observe that the polynomial $\sigma^t(f)$ maps to
$\bT^t(f)$ under the difference algebra homomorphism $\varphi$.
\end{proof}

A possible problem with the equations satisfied by the keys compatible with a given keystream
is that they could have a very high degree if the explicit difference system $(\ref{dsys})$
is non-linear. To avoid this problem, one can perform a partial elimination only with
the lowest degree equations of the system $(\ref{dsys})$. We make use of this approach
when attacking the cipher \E0\ in the next sections since its explicit difference system
contains linear equations, that is, LFSRs. We formalize this partial elimination approach
by the following results which are a straightforward generalization of the previous ones.

Fix an integer $0\leq m\leq n$. Denote
\begin{equation*}
\begin{gathered}
\bX'_m = \{x_1(0),\ldots,x_1(r_1-1),\ldots,x_m(0),\ldots,x_m(r_m-1)\}, \\
\bX''_m = \bigcup_{t\geq 0} \{x_{m+1}(t),\ldots,x_n(t)\}
\end{gathered}
\end{equation*}
and put $\bX_m = \bX'_m\cup\bX''_m\subset X$. Denote by $\bR_m = \KK[\bX_m]\subset R$
the corresponding subalgebra. We define the endomorphism $\bT_m:\bR_m\to\bR_m$ such that
\begin{equation*}
\begin{gathered}
x_i(0)\mapsto x_i(1),\ldots,x_i(r_i-2)\mapsto x_i(r_i-1),
x_i(r_i-1)\mapsto f_i\ (1\leq i\leq m), \\
x_i(t)\mapsto x_i(t+1)\ (m+1\leq i\leq n).
\end{gathered}
\end{equation*}
We have clearly that $\bT_0 = \sigma$ and $\bT_n = \bT$. Define the ideal $\bI_m\subset\bR_m$
such that
\[
\bI_m = \sum_{t\geq 0} \langle x_{m+1}(r_{m+1}+t) - \bT_m^t(f_{m+1}),\ldots,
x_n(r_n+t) - \bT_m^t(f_n) \rangle.
\]
Note that $\bT_m(\bI_m)\subset \bI_m$ and
$\bI_m = \langle x_{m+1}(r_{m+1}) - f_{m+1}, \ldots, x_n(r_n) - f_n \rangle_{\bT_m}$.
Denote by $\bT'_m:\bR_m/\bI_m\to\bR_m/\bI_m$ the algebra endomorphism which is induced by
$\bT_m:\bR_m\to\bR_m$.

\begin{theorem}
Let $\varphi_m:R\to\bR_m$ be the algebra homomorphism such that its restriction to $\bR_m$
is the identity map and for all $t\geq 0$ one has that
\[
x_i(r_i + t)\mapsto \bT_m^t(f_i)\ (1\leq i\leq m).
\]
It holds that $\varphi_m$ is a difference algebra homomorphism and its kernel is
$I = \ker\varphi_m$.
\end{theorem}

\begin{corollary}
The difference algebras $(R/I,\sigma')$ and $(\bR_m/\bI_m,\bT'_m)$ are isomorphic by means
of the difference algebra isomorphism $\varphi'_m:R/I\to\bR_m/\bI_m$ which is induced by
$\varphi_m$.
\end{corollary}

In terms of corresponding polynomial maps, we have the following result.

\begin{proposition}
\label{parelim1}
Consider $V_\KK(\bI_m)$ the set of all $\KK$-solutions of the system
of polynomial equations
\begin{equation*}
\left\{
\begin{array}{rcl}
x_{m+1}(r_{m+1} + t) & = & \bT_m^t(f_{m+1}) \\
& \vdots \\ 
x_n(r_n + t) & = & \bT_m^t(f_n) \\
\end{array}
\right.
\quad
(t\geq 0)
\end{equation*}
over the set of variables $\bX_m$. The map $\iota_m:V_\KK(I)\to V_\KK(\bI_m)$ such that
\[
(a_1,\ldots,a_n)\mapsto
(a_1(0),\ldots,a_1(r_1-1),\ldots,a_m(0),\ldots,a_m(r_n-1),a_{m+1},\ldots,a_n)
\]
is bijective and $\iota_m,\iota_m^{-1}$ are both polynomial maps.
\end{proposition}

\begin{proposition}
\label{parelim2}
We have that $\iota_m(V_\KK(I + J)) = V_\KK(\bI_m + \bJ_m)$ where
\[
\bJ_m = \sum_{t\geq T} \langle \bT_m^t(f) - b(t) \rangle\subset \bR_m
\]
\end{proposition}

Note that actual computations with the ideal $\bJ_m$ are done by assuming that
$T\leq t\leq B$ for a sufficiently large bound $B\geq T$. Moreover, for an invertible
cipher we always assume that $T = 0$.

Finally, it is useful to observe the following. Suppose that $R$ is endowed with a monomial
ordering $\succ$ such that, for all $t\geq 0$
\[
x_i(r_i + t)\succ \sigma^t(f_i)\ (1\leq i\leq m).
\]
In this case, for any polynomial $g\in\bR_m$ we can compute $\bT_m^t(g)$ in an alternative way.
Consider the ideal
\begin{equation*}
\begin{gathered}
I_m = \sum_{t\geq 0}\langle x_1(r_1 + t) - \sigma^t(f_1),\ldots,
x_m(r_m + t) - \sigma^t(f_m) \rangle \\
= \langle x_1(r_1) - f_1,\ldots, x_m(r_m) - f_m \rangle_\sigma\subset R
\end{gathered}
\end{equation*}
whose generating set is a \Gr\ basis because its leading monomials are distinct
linear ones. We have that $\bT_m^t(g)$ is indeed the normal form of the polynomial
$g$ modulo $I_m$ with respect to the given monomial ordering. For more details,
see \cite{LST}.

\section{The stream cipher \E0}

From now on, we assume that the finite base field is $\KK = \GF(2) = \ZZ_2$.
In order to avoid confusion, in the present section we denote by $\dplus$
the sum in $\ZZ$ and by $+$ the sum in $\KK$.
The Bluetooth stream cipher \E0 is obtained by four Linear Feedback Shift Registers
(briefly LFSRs) that are combined by a Finite State Machine (FSM) as described
in Figure 1.

\begin{figure}[H]
\begin{center}
\includegraphics[width=1.0\textwidth]{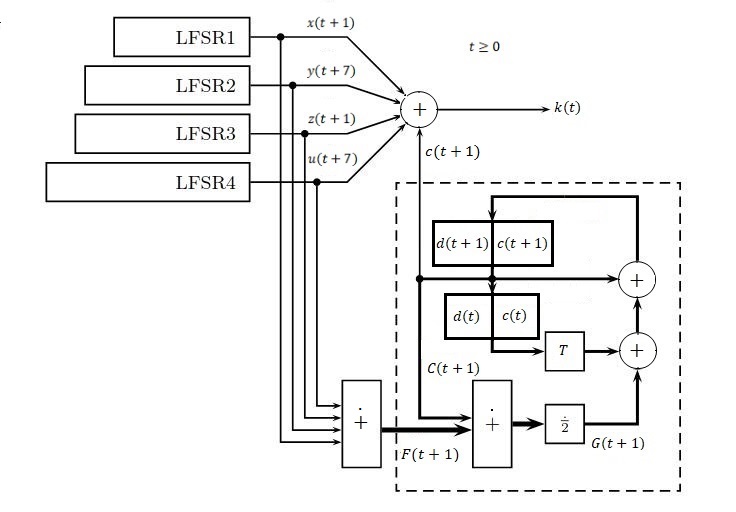}
\caption{\E0\ Bluetooth Stream Cipher}
\end{center}
\end{figure} 

The four LFSRs are respectively defined by the following primitive polynomials
with coefficients in $\KK$
\[
\begin{array}{rcl}
p_1 & = & x^{25} + x^{17} + x^{13} + x^5 + 1, \\
p_2 & = & x^{31} + x^{19} + x^{15} + x^7 + 1, \\
p_3 & = & x^{33} + x^{29} + x^9 + x^5 + 1, \\
p_4 & = & x^{39} + x^{35} + x^{11}+ x^3 + 1. \\
\end{array}
\]

\medskip
For all clocks $t\geq 0$, the state of the FSM consists of $4$ bits which are stored
in a pair of 2-bit delay elements, say
\[
(d(t),c(t)), (d(t+1),c(t+1))\in\KK^2.
\]
It is useful to define the corresponding integer numbers
\begin{equation*}
\begin{gathered}
C(t) = d(t) 2 \dplus c(t), \hspace{18pt} \\
C(t+1) = d(t+1) 2 \dplus c(t+1).
\end{gathered}
\end{equation*}
At any clock, the lower delay element stores the previous value
of the upper element, that is, $(d(t+1),c(t+1))$ stores in $(d(t),c(t))$.
Then, the new 2-bit $(d(t+2),c(t+2))$ for the upper delay element of the combiner
is computed by putting
\[
(d(t+2),c(t+2)) = (g_1(t+1),g_0(t+1)) + (d(t+1),c(t+1)) + T(d(t),c(t))
\]
where $T:\KK^2\rightarrow \KK^2$ is the linear bijection
\[
(d(t),c(t))\mapsto (c(t), d(t) + c(t))
\]
and the 2-bit $(g_1(t+1),g_0(t+1))\in\KK^2$ is defined as follows. Consider the sum
\[
F(t) = x(t)\dplus y(t+6)\dplus z(t)\dplus u(t+6)\in \ZZ
\]
and define the integer
\[
G(t+1) = \left\lfloor \frac{F(t+1)\dplus C(t+1)}{2} \right\rfloor.
\]
Since $0\leq G(t+2)\leq 3$, we define the 2-bit element $(g_1(t+1),g_0(t+1))\in\KK^2$
as the binary representation of $G(t+1)$, namely
\[
G(t+1) = g_1(t+1) 2\dplus g_0(t+1).
\]
Finally, for all $t\geq 0$ the keystream bits of the cipher \E0 are
computed as the sum
\[
k(t) = x(t+1) + y(t+7) + z(t+1) + u(t+7) + c(t+1)\in \KK.
\]

\section{\E0 is a difference stream cipher}

We show now that we can obtain \E0 as a difference stream cipher, that is, we can translate it
into a system of difference equations. The four independent LFSRs correspond immediately to
the following subsystem of linear difference equations

\begin{equation}
\left\{
\begin{array}{rcl}
x(25) & = & x(0) + x(5) + x(13) + x(17), \\
y(31) & = & y(0) + y(7) + y(15) + y(19), \\
z(33) & = & z(0) + z(5) + z(9) + z(29), \\
u(39) & = & u(0) + u(3) + u(11) + u(35), \\
\end{array}
\right.
\end{equation}

Let us consider now the FSM combiner. Since $0\leq F(t)\leq 4$, the binary representation
of $F(t)$ consists of a 3-bit element $(f_2(t),f_1(t),f_0(t))\in\KK^3$, that is
\[
F(t) = f_2(t) 2^2\dplus f_1(t) 2\dplus f_0(t).
\]
Clearly, we can view $F(t) = x(t)\dplus y(t+6)\dplus z(t)\dplus u(t+6)$ as a function
$\KK^4\to \ZZ$ with variable set $\{x(t), y(t+6), z(t), u(t+6)\}$ and therefore
$f_0(t),f_1(t),f_2(t)$ as Boolean functions $\KK^4\to \KK$ with the same variable set.
By converting the latter functions into their Algebraic Normal Form (briefly ANF), that is,
as elements of the polynomial algebra $\KK[x(t),y(t+6),z(t),u(t+6)]$ modulo the identities
\begin{equation*}
\begin{gathered}
x(t)^2 + x(t) = y(t+6)^2 + y(t+6) = z(t)^2 + z(t) = u(t+6)^2 + u(t+6) = 0
\end{gathered}
\end{equation*}
we obtain that
\[
\begin{array}{rcl}
f_0(t) & = & x(t) + y(t+6) + z(t) + u(t+6), \\
f_1(t) & = & x(t)y(t+6) + x(t)z(t) + x(t)u(t+6) + y(t+6)z(t) \\
       &   & +\, y(t+6)u(t+6) + z(t)u(t+6), \\
f_2(t) & = & x(t) y(t+6) z(t) u(t+6).
\end{array}
\]
Standard methods to obtain the ANF involves Support and Minterm Representation of Boolean
functions and we refer to the book \cite{WF} for all details about them.

We consider now the following sum of integers
\[
F(t)\dplus C(t) = f_2(t) 2^2\dplus f_1(t) 2\dplus f_0(t)\dplus d(t) 2\dplus c(t)
\]
whose binary representation, as a result of carries, is
\begin{equation*}
\begin{gathered}
F(t)\dplus C(t) =  (f_2(t) + f_1(t)d(t) + f_0(t)c(t) d(t) + f_1(t)f_0(t)c(t)) 2^2 \\
\dplus\, (f_0(t)c(t) + f_1(t) + d(t)) 2\dplus (f_0(t) + c(t)).
\end{gathered}
\end{equation*}

By dividing by 2, we obtain that
\begin{equation*}
\begin{gathered}
\frac{F(t)\dplus C(t)}{2} =  (f_2(t) + f_1(t)d(t) + f_0(t)c(t) d(t)
+ f_1(t)f_0(t)c(t)) 2 \\
\dplus\, (f_0(t)c(t) + f_1(t) + d(t))\dplus (f_0(t) + c(t)) 2^{-1}.
\end{gathered}
\end{equation*}
and therefore it holds that
\begin{equation*}
\begin{gathered}
\left\lfloor \frac{F(t)\dplus C(t)}{2} \right\rfloor =
(f_2(t) + f_1(t)d(t) + f_0(t)c(t)d(t) + f_1(t)f_0(t)c(t)) 2 \\
\dplus\, (f_0(t)c(t) + f_1(t) + d(t)).
\end{gathered}
\end{equation*}

Since, by definition, we have that
\[
G(t+1) = g_1(t+1) 2\dplus g_0(t+1) = \left\lfloor \frac{F(t+1)\dplus C(t+1)}{2} \right\rfloor
\]
we conclude that
\[
\begin{array}{rcl}
g_1(t+1) & = & f_2(t+1) + f_1(t+1) d(t+1) + f_0(t+1) c(t+1) d(t+1) \\
         &   & +\, f_1(t+1) f_0(t+1) c(t+1), \\
g_0(t+1) & = & f_0(t+1) c(t+1) + f_1(t+1) + d(t+1). \\
\end{array}
\]
Finally, by using the definition
\[
(d(t+2),c(t+2)) = (g_1(t+1),g_0(t+1)) + (d(t+1),c(t+1)) + (c(t), d(t) + c(t))
\]
we obtain that
\[
\begin{array}{rcl}
d(t+2) & = & f_2(t+1) + f_1(t+1) d(t+1) + f_0(t+1) c(t+1) d(t+1) \\
       &   & +\, f_1(t+1) f_0(t+1) c(t+1) + d(t+1) + c(t), \\
c(t+2) & = & f_0(t+1) c(t+1) + f_1(t+1) + d(t+1) + c(t+1) + d(t) + c(t). \\
\end{array}
\]

\medskip
Define now the variable set
\begin{equation*}
\begin{gathered}
\bX = \{x(0),\ldots,x(24),y(0),\ldots,x(30),z(0),\ldots,z(32), \\
                        u(0),\ldots,u(38),c(0),c(1),d(0),d(1)\}
\end{gathered}
\end{equation*}
and the corresponding polynomial algebra $\bR = \KK[\bX]$. By the above calculations,
we have that \E0 is a difference stream cipher whose evolution of the internal state
$v(t)\in\KK^{132}$ ($t\geq 0$) is described by the following system of explicit difference
equations
\begin{equation}
\label{e0sys}
\left\{
\begin{array}{rcl}
x(25) & = & x(0) + x(5) + x(13) + x(17), \\
y(31) & = & y(0) + y(7) + y(15) + y(19), \\
z(33) & = & z(0) + z(5) + z(9) + z(29), \\
u(39) & = & u(0) + u(3) + u(11) + u(35), \\
c(2) & = & g'_0, \\
d(2) & = & g'_1 \\
\end{array}
\right.
\end{equation}
where the non-linear polynomials $g'_0,g'_1\in \bR$ are defined as
\[
\begin{array}{rcl}
g'_0 & = & x(1) c(1) + y(7) c(1) + z(1) c(1) + u(7) c(1) + x(1) y(7) + x(1) z(1) \\
     &   & +\, x(1) u(7) + y(7) z(1) + y(7) u(7) + z(1) u(7) + c(1) + d(1) + c(0) + d(0), \\
g'_1 & = & x(1) y(7) z(1) u(7) + x(1) y(7) d(1) + x(1) z(1) d(1) + x(1) u(7) d(1) \\
     &   & +\, y(7) z(1) d(1) + y(7) u(7) d(1) + z(1) u(7) d(1) + x(1) c(1) d(1) \\
     &   & +\, y(7) c(1) d(1) + z(1) c(1) d(1) + u(7) c(1) d(1) + x(1) y(7) z(1) c(1) \\
     &   & +\, x(1) y(7) u(7) c(1) + x(1) z(1) u(7) c(1) + y(7) z(1) u(7) c(1) + d(1) + c(0). \\
\end{array}
\]
Finally, the keystream polynomial $f\in\bR$ of \E0 is defined as
\[
f = x(1) + y(7) + z(1) + u(7) + c(1).
\]
In other words, for each clock $t\geq 0$ a keystream of \E0 is obtained by evaluating
the polynomial $f$ over the $t$-state $v(t)\in\KK^{132}$ of a $\KK$-solution of (\ref{e0sys}).

Note that Bluetooth specifications \cite{Bl} require that the keystream outputs after
a reinitialization at clock $T = 200$. Finally, remark that our description of \E0 matches
with the sample data contained in Appendix IV of those specifications once all initial states
of the LFSRs are reversed and one considers the initial state of the FSM as $(c(0),d(0),c(1),d(1))$.
Indeed, the initial states of the non-linear equations of (\ref{e0sys}) are $(c(0),c(1))$
and $(d(0),d(1))$.

\section{An algebraic attack to \E0}

As a result of Theorem \ref{invth} and Definition \ref{invsys}, we have that the explicit
difference system (\ref{e0sys}) of \E0 is invertible with inverse system
\begin{equation*}
\left\{
\begin{array}{rcl}
x(25) & = & x(20) + x(12) + x(8) + x(0), \\
y(31) & = & y(24) + y(16) + y(12) + y(0), \\
z(33) & = & z(28) + z(24) + z(4) + z(0), \\
u(39) & = & u(36) + u(28) + u(4) + u(0), \\
c(2) & = & h_0, \\
d(2) & = & h_1 \\
\end{array}
\right.
\end{equation*}
where the polynomials $h_0,h_1\in\bR$ are defined as
\[
\begin{array}{rcl}
h_0 & = & x(24) y(24) z(32) u(32) + x(24) y(24) z(32) c(1) + x(24) y(24) u(32) c(1) \\
& & +\, x(24) z(32) u(32) c(1) + y(24) z(32) u(32) c(1) + x(24) y(24) d(1) \\
& & +\, x(24) z(32) d(1) + y(24) z(32) d(1) + x(24) u(32) d(1) + y(24) u(32) d(1) \\
& & +\, z(32) u(32) d(1) + x(24) c(1) d(1) + y(24) c(1) d(1) + z(32) c(1) d(1) \\
& & + u(32) c(1) d(1) + d(1) + d(0), \\
h_1 & = & x(24) y(24) z(32) u(32) + x(24) y(24) z(32) c(1) + x(24) y(24) u(32) c(1) \\
& & +\, x(24) z(32) u(32) c(1) + y(24) z(32) u(32) c(1) + x(24) y(24) d(1) \\
& & +\, x(24) z(32) d(1) + y(24) z(32) d(1) + x(24) u(32) d(1) + y(24) u(32) d(1) \\
& & +\, z(32) u(32) d(1) + x(24) c(1) d(1) + y(24) c(1) d(1)+ z(32) c(1) d(1) \\
& & +\, u(32) c(1) d(1) + x(24) y(24) + x(24) z(32) + y(24) z(32) + x(24) u(32) \\
& & +\, y(24) u(32) + z(32) u(32) + x(24) c(1) + y(24) c(1) + z(32) c(1) \\
& & +\, u(32) c(1) + c(1) + c(0) + d(0). \\
\end{array}
\]
We recall that this inverse system is easily obtained by computing a suitable
\Gr\ basis. The invertibility of the system (\ref{e0sys}) allows us to attack
equivalently any internal state. A convenient choice consists hence in attacking
the state corresponding to the clock where the keystream starts to output.

With the notations of the Section 3, in our experiments we choose to compute
$V_\KK(I_B + J_B)$ for a clock bound $B$ in the range $57\leq B\leq 69$, that is,
we use for the attack the knowledge of a number $K$ of keystream bits in the range
$51\leq K\leq 63$. To reduce the number of tests, we consider $K$ odd.

In this range we have found very few $\KK$-solutions at each instance, in a suitably
large test set, of a guess-and-determine strategy for solving $V_\KK(I_B + J_B)$.
This strategy is based on the exhaustive evaluations of 83 variables and in the
next section we explain how 14 of such variables have been chosen to speed-up
the computations. The considered test set consists of $2^{17}$ evaluations
for $2^3$ different keys.

For each guess of the 83 variables, we are able to determine the number
of $\KK$-solutions of the corresponding polynomial system as the $\KK$-dimension 
of the quotient algebra modulo the ideal generated by the system and the field
equations (see Section 2). Such dimension is easily obtained as a by-product
of the DegRevLex-\Gr\ basis that is computed at each evaluation.

These $\KK$-solutions for wrong guesses of the 83 variables corresponds to
spurious keys which can be detected by using some additional very small number
of values of the keystream. Indeed, for $K = 63$ the number of spurious keys
for each guess is on the average close to zero. Note that a good trade-off consists in using
a value of $K$ that lies approximately in the middle of the range $51\leq K\leq 63$
because the cost of solving grows significantly with $K$ but the cost of computing
and comparing a few extra keystream bits is indeed very small.

Since the explicit difference system of \E0\ contains the combiner equations
which are of degree 2 and 4, the approach we use to define the polynomial system
to solve is the partial elimination described in Proposition \ref{parelim1}
and Proposition \ref{parelim2}.
Namely, we perform elimination by the linear difference equations, that is,
the LFSRs of the system $(\ref{e0sys})$. Once we fix a clock bound $B$,
the corresponding polynomial system is defined over the following variable set
\begin{equation*}
\begin{gathered}
\{x(0),\ldots,x(24),y(0),\ldots,y(30),z(0),\ldots,z(32), \\
u(0),\ldots,u(38),c(0),\ldots,c(B),d(0),\ldots,d(B)\}.
\end{gathered}
\end{equation*}
In our guess-and-determine strategy, a subset of 83 such variables are evaluated
over the finite field $\KK = \GF(2)$ in an exhaustive way, leading to \Gr\ bases computations
that take few tens of milliseconds on the average.

We observe that other partial eliminations could be considered,
including total elimination of all variables except for the initial ones
that are
\begin{equation*}
\begin{gathered}
\{x(0),\ldots,x(24),y(0),\ldots,y(30),z(0),\ldots,z(32), \\
u(0),\ldots,u(38),c(0),c(1),d(0),d(1)\}.
\end{gathered}
\end{equation*}
Indeed, we have experimented that all these variants increase the degree of the
eliminated polynomials in a way that either makes them impossible to be computed
or leads to polynomial systems which are more difficult to solve.

\section{Fourteen useful variables}

The set of the 83 evaluated variables that we have used to attack \E0 by means of
a guess-and-determine strategy applied to its difference stream cipher structure
is the following one
\begin{equation*}
\begin{gathered}
\{
x(0),\ldots,x(24),y(0),\ldots,y(26),y(29),z(0),\ldots,z(10),z(29),\ldots,z(32), \\
u(0),\ldots,u(9),u(35),u(36),u(37),c(0),d(0)
\}.
\end{gathered}
\end{equation*}
Fourteen of the above variables have been single out by means of the arguments in this section.
The remaining 69 variables have been obtained by an experimental optimization.

The monomial ordering that we use for computing \Gr\ bases and normal forms during the attack
to \E0 is defined as the DegRevLex-ordering over the following variable set
\begin{equation*}
\begin{gathered}
\{x(0),y(0),z(0),u(0), x(1),y(1),z(1),u(1),\ldots\} \\
\cup\, \{c(0), c(1), \ldots\}\cup \{d(0), d(1), \ldots\}
\end{gathered}
\end{equation*}
induced by putting
\begin{equation*}
\begin{gathered}
x(0)\prec y(0)\prec z(0)\prec u(0)\prec x(1)\prec y(1)\prec z(1)\prec u(1)\prec \ldots \\
\prec c(0)\prec c(1)\prec \ldots \prec d(0)\prec d(1)\prec \ldots
\end{gathered}
\end{equation*}

We consider the following polynomials which belong to the difference ideal $I$ corresponding to
the explicit difference system $(\ref{e0sys})$
\[
\begin{array}{rcl}
C_0 & = & c(2) + x(1) c(1) + y(7) c(1) + z(1) c(1) + u(7) c(1) + x(1) y(7) + x(1) z(1) \\
    &   & +\, x(1) u(7) + y(7) z(1) + y(7) u(7) + z(1) u(7) + c(1) + d(1) + c(0) + d(0), \\

C_1 & = & c(3) + x(2) c(2) + y(8) c(2) + z(2) c(2) + u(8) c(2) + x(2) y(8) + x(2) z(2) \\
    &   & +\, x(2) u(8) + y(8) z(2) + y(8) u(8) + z(2) u(8) + c(2) + d(2) + c(1) + d(1), \\

D_0 & = & d(2) + x(1) y(7) z(1) u(7) + x(1) y(7) d(1) + x(1) z(1) d(1) + x(1) u(7) d(1) \\
    &   & +\, y(7) z(1) d(1) + y(7) u(7) d(1) + z(1) u(7) d(1) + x(1) c(1) d(1) \\
    &   & +\, y(7) c(1) d(1) + z(1) c(1) d(1) + u(7) c(1) d(1) + x(1) y(7) z(1) c(1) \\
    &   & +\, x(1) y(7) u(7) c(1) + x(1) z(1) u(7) c(1) + y(7) z(1) u(7) c(1) + d(1) + c(0). \\
\end{array}
\]
The above polynomials clearly arise from the combiner equations. Note that these polynomials
are in normal form with respect to the linear polynomials in $I$ corresponding to the LFRSs.
We also consider the following polynomials corresponding to the first 3 keystream bits,
say $b_0,b_1,b_2\in\KK$, of \E0
\[
\begin{array}{rcl}
B_0 & = & x(1) + y(7) + z(1) + u(7) + c(1) + b_0, \\
B_1 & = & x(2) + y(8) + z(2) + u(8) + c(2) + b_1, \\
B_2 & = & x(3) + y(9) + z(3) + u(9) + c(3) + b_2.
\end{array}
\]
These polynomials belong to the ideal $J$ that imposes to the $\KK$-solutions of $I$ to be
compatible with a given keystream (see Section 3). Note that $B_0,B_1,B_2$ are also in normal
form with respect to the LFSRs. Before computing the $\KK$-solutions of $I + J$ by means
of a \Gr\ basis, we can perform the normal form of $C_0,C_1,D_0$ modulo $B_0,B_1,B_2$
in order to eliminate the variables $c(1),c(2),c(3)$. These normal forms are the polynomials
\begin{equation*}
\begin{array}{rcl}
G_1 & = & d(1) + A_1, \\
G_2 & = & d(1) + d(2) + A_2, \\
G_3 & = & A_3 d(1) + d(2) + A_4
\end{array}
\end{equation*}
where
\begin{equation*}
\begin{array}{rcl}
A_1 & = & u(7) x(1) + u(7) y(7) + u(7) z(1) + x(1) y(7) + x(1) z(1) + y(7) z(1)  \\
    &   & +\, b_0 (x(1) + y(7) + z(1) + u(7)) + c(0) + d(0) + x(2) + y(8) + z(2) \\
    &   & +\, u(8) + b_0 + b_1, \\
A_2 & = &  u(8) x(2) + u(8) y(8) + u(8) z(2) + x(2) y(8) + x(2) z(2) + y(8) z(2) \\
    &   & +\, b_1 (x(2) + y(8) + z(2) + u(8)) + x(1) + x(3) + y(7) + y(9) + z(1) \\
    &   & +\, z(3) + u(7) + u(9) + b_0 + b_1 + b_2, \\
A_3 & = & u(7) x(1) + u(7) y(7) + u(7) z(1) + x(1) y(7) + x(1) z(1) + y(7) z(1) \\
    &   & +\, (b_0 + 1) (x(1) + y(7) + z(1) + u(7)) + 1, \\
A_4 & = & u(7) x(1) y(7) z(1) + (b_0 + 1) (u(7) x(1) y(7) + u(7) x(1) z(1) \\
    &   & +\, u(7) y(7) z(1) + x(1) y(7) z(1)) + c(0). \\
\end{array}
\end{equation*}
It is clear that the linear equations $G_1 = G_2 = G_3 = 0$ in the variables $d(1),d(2)$
are inconsistent if and only if
\[
G = (A_1 + 1) A_3 + A_2 + A_4 \neq 0.
\]
Note now that the set of $\KK$-solutions of the equation $G = 0$ is a preimage
of the Boolean function $\KK^{14}\to\KK$ corresponding to the polynomial $G$
in the 14 variables
\begin{equation*}
\begin{gathered}
\{x_1,x_2,x_3, y_7,y_8,y_9, z_1,z_2,z_3, u_7,u_8,u_9, c_0,d_0\}.
\end{gathered}
\end{equation*}
By computing the $\KK$-dimension of the quotient algebra
\[
\KK[x_1,\ldots,d_0]/\langle G, x_1^2 + x_1, \ldots, d_0^2 + d_0 \rangle
\]
we have that the number $\KK$-solutions of $G = 0$ is exactly $2^{13}$, for all bits
$b_0,b_1,b_3\in\KK$. In other words, the Boolean function corresponding to $G$ is
a balanced one, that is, its two preimages have the same number of elements.
Since $G_1,G_2,G_3$ are linear polynomials, this implies that the computation
of the \Gr\ bases of the guess-and-determine strategy is very fast for half
of the evaluations of the considered 14 variables. We can possibly precompute
the $\KK$-solutions of the equation $G = 0$ once given the first 3 keystream
bits $b_0,b_1,b_2$ in order to avoid useless \Gr\ bases computations.
In fact, in the experiments of the next section we perform
\Gr\ bases for all the evaluations of the 14 variables because they are extremely
fast in the case that $G\neq 0$ and one needs the evaluations of $69$ additional
variables to obtain fast computations also in the case $G = 0$.

Remark finally that in our attack to \E0, before performing \Gr\ bases computations,
we always eliminate also the variables $c(t+1)$ ($t\geq 0$) by means of the polynomials
\[
\sigma^t(f) + b(t) = x(t+1) + y(t+7) + z(t+1) + u(t+7) + c(t+1) + b(t)
\]
where $b(t)$ denotes the keystream bit at clock $t$.

\section{Experimental results}

In this section we report the results of our testing activity. Firstly, we code
an algebraic attack to the difference stream cipher \E0 using \Gr\ bases, SAT solvers
and Binary Decision Diagrams. Secondly, we run it on a couple of servers where
the second one is used only to allow parallel computations with large memory consumption
for BDDs. The servers have the following hardware configurations:
\begin{itemize}
\item Intel(R) Core(TM) i9-10900 CPU@2.80GHz, 10 Cores, 20 Threads and 64 Gb of RAM
 --- server A, for short;
\item 2 x Intel(R) Xeon(R) Gold 6258R CPU@2.7GHz, 56 Cores, 112 Threads and 768 Gb of RAM
 --- server B, for short.
\end{itemize}
On both these machines we install a Debian-based Linux distribution as operating system.

As described in the previous sections, for both \Gr\ bases and SAT solvers we make use
of a guess-and-determine strategy based on the evaluation of 83 variables and the
knowledge of a small number $K$ of keystream bits in the range $51\leq K\leq 63$.
Note that Bluetooth specifications require a reinitialization of the initial state
of the explicit difference system $(\ref{e0sys})$ by means of the last 128 keystream
bits obtained in the first 200 clocks. Because our algebraic attack can be performed
using less than 128 keystream bits, a full attack to \E0 is achieved by simply running
our code twice.

To compare \Gr\ bases with SAT solvers, we execute $2^{20}$ different tests on the
server A. More precisely, we consider $2^{17}$ random guesses of the 83 variables
and we use $2^3$ random keys. For each number $K$ of keystream bits, we gather average,
min and max computing times for performing DegRevLex-\Gr\ bases and SAT solving
and we report these data in Table \ref{gbvssat}. The timings are expressed
in milliseconds that are denoted as ``ms''. The chosen \Gr\ bases implementation
for our testing activity is {\sc slimgb} of the computer algebra system {\sc Singular}
\cite{DGPS} and the considered SAT solver is {\sc cryptominisat} \cite{So}.
Note that SAT solving is applied to the the same polynomial systems where
\Gr\ bases are computed once these systems are converted in the Conjunctive
Normal Form (briefly CNF). Note that the ANF-to-CNF conversion time is essentially
negligible since we apply this transform only once and for each evaluation
of the 83 variables we just add the corresponding linear equations to the CNF.

\begin{table}[H]
\centering
\caption{GB vs SAT}
\label{gbvssat}
\vspace{4pt}
\footnotesize
\begin{tabular}{|c|c|c|c|c|}
\hline
$K$ & GB avg & GB min/max & SAT avg & SAT min/max  \\
\hline			               
51 & 31ms   & 1/411ms    & 196ms   & 105/1007ms   \\
\hline			               
53 & 34ms   & 2/480ms    & 220ms   & 121/876ms    \\
\hline			               
55 & 41ms   & 2/522ms    & 230ms   & 134/638ms    \\
\hline			               
57 & 52ms   & 3/620ms    & 245ms   & 143/645ms    \\
\hline			               
59 & 64ms   & 3/799ms    & 283ms   & 161/777ms    \\
\hline			               
61 & 79ms   & 3/1115ms   & 300ms   & 174/732ms    \\
\hline			               
63 & 96ms   & 4/1287ms   & 326ms   & 191/862ms    \\
\hline
\end{tabular}
\end{table}

According to Section 7, all minimum computing times are actually obtained for guesses
of the 14 special variables such that $G\neq 0$.

\begin{table}[H]
\centering
\caption{GB data}
\label{gbdata}
\vspace{4pt}
\footnotesize
\begin{tabular}{|c|c|c|c|c|c|c|c|c|c|}
\hline
$K$ & deg(GB)$=$0 & deg(GB)$=$1 & deg(GB)$=$2 & deg$=$1 avg sol & deg$=$2 avg sol \\
\hline
51 & 83.781\% & 15.243\% & 0.975\%  & 1.442 & 3.154 \\ 
\hline 
53 & 94.023\% & 5.971\%  & 0.005\%  & 1.047 & 3 \\
\hline 
55 & 98.438\% & 1.561\%  & 0.0001\% & 1.011 & 3 \\
\hline 
57 & 99.613\% & 0.386\%  & 0\%      & 1.004 & \\ 
\hline 
59 & 99.901\% & 0.098\%  & 0\%      & 1 & \\
\hline 		        
61 & 99.976\% & 0.023\%  & 0\%      & 1 & \\
\hline 
63 & 99.993\% & 0.006\%  & 0\%      & 1 & \\
\hline
\end{tabular}
\end{table}

Table \ref{gbdata} presents, for different values of $K$, the number of \Gr\ bases
of a certain degree and the average number of (spurious) solutions we compute by means
of such bases. We express the number of \Gr\ bases of some degree as a percentage
of the total number of \Gr\ bases in our test set which is $2^{20}$ for any $K$.
The degree of a \Gr\ basis is the highest degree of its elements up to field equations.
A \Gr\ basis of degree 0 corresponds to an inconsistent polynomial system, that is,
we have no spurious solutions. \Gr\ bases of degree strictly greater than 2 were not
found in our tests.

Data gathered show that the average number of spurious solutions for each
\Gr\ basis drops down very quickly as the number $K$ of keystream bits
slightly increases. If we set $K\geq 59$, more than $99.9\%$ of the \Gr\ bases
provide no spurious solution. The remaining $0.1\%$ consists of \Gr\ bases
of degree 1 with a single solution. Such a solution can be read immediately
from the basis and detected as a spurious one by using few additional
keystream bits. In fact, for $K = 63$ the probability to have spurious
solutions is close to zero.


\smallskip \noindent
In order to validate the timings collected using \Gr\ bases and SAT solvers, we also code
a BDD-based algebraic attack to \E0 and compare new results with those presented
in Table \ref{gbvssat} and in the literature \cite{SGPS,SW}. Indeed, BDDs have been generally
considered the standard in \E0 cryptanalysis. We install \BuDDy\ library package 2.4 \cite{LN}
on our machines and following the approach described in Section 3 of \cite{SGPS},
we generate a number $K$ of BDDs (consisting of unknown key variables) each of which
is associated with a Boolean equation. This set of $K$ equations corresponds to the
number $K$ of given keystream bits. This means that sometimes these equations are equal
to $0$, sometimes to $1$. In the former case, we take the complement of the Boolean equation,
whereas in the latter we do not. Now, we have several Boolean equations in unknown
key variables equate to $1$ and we have to find a common solution for these equations.
Such a solution can be found by ANDing our set of BDDs. Notice that AND operations are
usually extremely expensive both in time and memory, therefore the ordering to perform
ANDing is of fundamental importance. Among the various approaches described in \cite{SGPS}
such as sequential ANDing, ANDing with fixed interval, random ANDing, RSAND and so on,
we adopt RSAND because it takes the overall used memory under control, reducing
(recursively) the number of BDDs by half until it gets the final BDD. 

Now, we are able to conduct an extensive testing to gauge the performance of our
BDD-based attack. We initially consider the same set of 83 variables previously used
with \Gr\ bases and SAT solvers. More precisely, we set $57\leq K\leq 63$, collect
several random guesses of $83$ variables, use $2^3$ random keys and try to recover
the remaining $49$ key bits. Recall that the key for us is the 132-bit internal state
of \E0 at the clock where the keystream starts to output. 

Despite using RSAND, experimental activities show that none of these tests ended due
to lack of memory of our servers. Indeed, the running code requires more than 768 Gb of RAM
which is the maximum amount of memory available on our server B. Therefore, we reduce
the number of unknown key bits to be recovered from $49$ to $39$ and provide some random
evaluation of $93$ variables. Notice that the configuration with $39$ bits was recovered
in 5 seconds on a regular personal computer by the authors of \cite{SGPS}. On server A,
using a single-thread configuration and a few Mb of memory, we are able to recover $39$
unknown key bits in about 0.15 seconds. Interestingly, the set of variables which provide
better results is not the same found with \Gr\ bases and SAT solvers but it is a chunk
of consecutive bits which includes all $39$ variables of the fourth LFSR, namely
$u(0),\ldots, u(38)$. If we increase the number of variables to be solved from $39$
to $40, 41, 42$ and so on, our testing activities suggest to include last variables
of the third LFSR, that is, $z(32), z(31), z(30), \dots$. In particular,
we have experimentally verified that a different chunk of variables, as well as several
variations (consecutive and not), yields worse timings.

We then measure the performance of the better chunk which is identified as
\[
u(0), \ldots, u(38), z(32), z(31), z(30), \ldots
\]
and starting from $39$ variables we increase this set by one variable at each time.
Table \ref{test_buddy} shows the results of the experimental activity with \BuDDy\ library 2.4
on server A which is slightly faster than server B. The timings are given in seconds that
will be denoted as ``s''.

\begin{table}[h]
\centering
\caption{BDD with \BuDDy\ 2.4}
\label{test_buddy}
\vspace{4pt}
\footnotesize
\begin{tabular}{|c|c|c|c|c|}
\hline
key bits & $K$ & exec time & mem used & \# of threads\\
\hline			               
39       & 40  & 0.15s  &   60Mb & 1 \\
\hline			               
40       & 41  & 1.07s  &   240Mb & 1 \\
\hline			               
41       & 42  & 4.75s  &   725Mb & 1 \\
\hline			               
42       & 43  & 19.3s  &   3Gb   & 1 \\
\hline			               
43       & 44  & 94.5s  &   13Gb  & 1 \\
\hline
44 & 45  & -- &  out of mem  & 1 \\
\hline
\end{tabular}
\end{table}

Because elapsed time and memory used grow exponentially, and \BuDDy\ library does not
provide the possibility to run the code on all threads of our servers, we install
\Sylvan\ \cite{DdP}, a decision diagram package which support multi-core architectures. 
The testing activities suggest that our code is slower on \Sylvan\ and faster on
\BuDDy\ 2.4 when executed in single thread mode. Notice that multiple factors can cause
getting speed results slower than the speed to which you are expected but this gap is
easily bridged by increasing the number of threads used. Exploiting the power of
the modern multi-core architectures, the advantage of \Sylvan\ becomes more and more
evident as the number of unknown key bits to recover increases.  

Again, we conducted an extensive testing to gauge the performance of our BDD-attack.
We set $K = 41, 43, 45$, collect $2^9$ random guesses of $91, 89, 87$ variables,
use $2^3$ random keys and we try to recover several bits of the key, collecting average
execution time and memory used. Table \ref{bdd_sylvan} summarizes the results of our
testing activities with \Sylvan\ on all 112 threads of server B. Notice that, due to
time consumption, last four rows of this table do not refer to $2^{12}$ different
tests --- $2^9$ random guesses and $2^3$ random keys --- but to a single execution
with a random key.  

\begin{table}[h]
\centering
\caption{BDD with \Sylvan}
\label{bdd_sylvan}
\vspace{4pt}
\footnotesize
\begin{tabular}{|c|c|c|c|c|}
\hline
key bits & $K$ & exec time & mem used & \# of threads\\
\hline			               
39 &  41 & 1.19s  & 1.47Gb & 112\\
\hline			               
40 &  41 & 1.55s  & 1.51Gb & 112\\
\hline			               
41 &  41 &  1.95s  & 1.60Gb  & 112\\

\hline			               
39 &  43 & 1.34s  & 1.60Gb & 112\\
\hline			               
40 &  43 & 2.03s  & 1.76Gb & 112\\
\hline			               
41 &  43 & 4.71s  & 3.46Gb  & 112\\

\hline			               
39 &  45 &  3.58s  & 3.44Gb & 112\\
\hline			               
40 &  45 &  5.23s  & 3.86Gb & 112\\
\hline			               
41 &  45 & 14.65s  & 7.41Gb  & 112\\
\hline
\hline
43 &  45 & 68.29s &  30Gb  & 112 \\
\hline
44 &  45 & 128.37s &  36Gb  & 112 \\
\hline
45 &  46 & 517.42s  &  233Gb  & 112 \\
\hline
46 &  47 & ---  &  out of mem  & 112 \\
\hline
\end{tabular}
\end{table}

\noindent
Our experimental results suggest that BDD-based algebraic attacks to \E0 are not up to
those obtained by using \Gr\ bases or SAT solvers. In addition to the huge difference
in computing times, note finally that all our \Gr\ bases computations run in less than
0.5 Gb of memory for $K = 63$.

\section{Conclusions}

This paper shows that the notion and theory of difference stream ciphers introduced
in \cite{LST} can be usefully applied to the algebraic cryptanalysis of realworld
ciphers as \E0 that is used in the Bluetooth protocol. In particular, the invertibility
property of the explicit difference system defining the evolution of the state of \E0\
allows to attack any internal state instead of the initial one, reducing computations
in a significant way. Moreover, the variables elimination obtained by the linear
difference equations corresponding to the LFSRs of \E0 contributes to improve
the performance of an algebraic attack. Finally, the difference stream cipher
structure of the Bluetooth encryption reveals that there are 14 special variables
which when evaluated, lead to linear equations among other variables. Such special variables
are useful then to speed-up a guess-and-determine strategy for solving the polynomial
system corresponding to the algebraic attack. Our attack is based on the exhaustive
evaluation of 83 state variables, including the 14 useful ones, and the knowledge
of about 60 keystream bits. We show that a low number of spurious keys are compatible 
with such short keystream which is a possible flaw of the cipher. The average solving
time by means of a \Gr\ basis of the polynomial system corresponding to each evaluation
is about 60 milliseconds. The sequential running time is hence about $2^{79}$
seconds by an ordinary CPU which improves any previous attempt to attack \E0 using
a short keystream. The complexity $2^{83}$ also improves the one obtained by BDD-based
cryptanalysis which is generally estimated as $2^{86}$ (see, for instance, \cite{Kl,SGPS}).
In fact, \Gr\ bases are compared in this paper with other solvers confirming
their feasibility in practical algebraic cryptanalysis already shown in \cite{LST}.
We finally observe that the parallelization of the brute force on the 83 variables
can be easily used to scale down further the runtime.

\section{Acknowledgements}

We would like to thank the anonymous referees for the careful reading of the
manuscript. We have sincerely appreciated all valuable comments and suggestions
as they have significantly improved the readability of the paper.

\end{document}